\documentclass[12pt,a4paper]{article}

\usepackage{epsf}
\usepackage{epsfig}
\usepackage{amsfonts}
\usepackage{amsmath}
\usepackage{amssymb}
\usepackage{amsthm}
\usepackage{setspace}
\usepackage{multirow}
\usepackage{graphics}
\usepackage{graphicx}
\usepackage{geometry}
\usepackage{array}
\usepackage{float}
\usepackage{lscape}
\usepackage{color}
\usepackage{calligra}
\usepackage{dsfont}
\usepackage{tabularx}
\usepackage{dcolumn}
\usepackage[bottom]{footmisc}
\usepackage[round,sort&compress]{natbib}
\usepackage{enumerate}
\usepackage{url}
\usepackage[normalem]{ulem}
\usepackage{epstopdf}
\usepackage[space]{grffile}
\usepackage{setspace}

\newtheorem{proposition}{Proposition}
\newtheorem{theorem}{Theorem}
\newtheorem{lemma}{Lemma}
\newtheorem{corollary}{Corollary}

\DeclareMathAlphabet{\mathpzc}{OT1}{pzc}{m}{it}
\DeclareMathAlphabet{\mathcalligra}{T1}{calligra}{m}{n}
\begin{document}

\newgeometry{textheight=25.1cm, voffset = -20pt}
\begin{titlepage}
\singlespacing

   \title{Assessment Voting in Large Electorates\thanks{We are grateful to Salvador Barber\`{a}, David Basin, Georgy Egorov, Lara Schmid and seminar participants at ETH Zurich for valuable discussions. All errors are our own.}}

	\author{Hans Gersbach \\
		\small CER-ETH -- Center of Economic \\
		\small Research at ETH Zurich and CEPR \\
		\small Z\"urichbergstrasse 18 \\
		\small 8092 Zurich, Switzerland \\
		\small hgersbach@ethz.ch \vspace{5mm}\\
	\and Akaki Mamageishvili \\
		\small CER-ETH -- Center of Economic \\
		\small Research at ETH Zurich  \\
		\small Z\"urichbergstrasse 18 \\
		\small 8092 Zurich, Switzerland \\
		\small amamageishvili@ethz.ch \vspace{5mm}\\
	\and Oriol Tejada \\
		\small CER-ETH -- Center of Economic \\
		\small Research at ETH Zurich  \\
		\small Z\"urichbergstrasse 18 \\
		\small 8092 Zurich, Switzerland \\
		\small toriol@ethz.ch \vspace{5mm}\\
	}
	\date{\normalsize First version: October 2016\\ \vspace{1mm}
			This version:  December 2017}
	\maketitle\thispagestyle{empty}

\vspace{-1cm}
\singlespacing

    \begin{abstract}
     We analyze Assessment Voting, a new two-round voting procedure that can be applied to binary decisions in democratic societies. In the first round, a randomly-selected {number} of citizens {cast their} vote on one of the two alternatives at hand, {thereby} irrevocably {exercising} their right to vote. In the second round, after the results of the first round have been published, the remaining citizens decide whether to vote for one alternative or to abstain. The votes from both rounds are aggregated, and the final outcome is obtained by applying the majority rule{, with ties being broken by fair randomization}. Within a costly voting framework, we show that large electorates will choose the preferred alternative of the majority with high probability, and that average costs will be low. This result is in contrast with the literature on one-round voting, which predicts either higher voting costs (when voting is compulsory) or decisions that often do not represent the preferences of the majority (when voting is voluntary).


\medskip 


\medskip


 \end{abstract}

\end{titlepage}

\restoregeometry

\onehalfspacing


	\setcounter{page}{2}

\section{Introduction}

How can the will of the majority in the {citizenry} be reflected in the outcomes of democratic decisions? Turnout {in elections tends to be significantly} lower than the size of the electorate and thus it is unclear whether {the citizens casting a vote can be trusted to represent} the distribution of preferences in the {entire population. As has been widely argued in the literature, some} citizens may not exercise {their right to vote when voting is costly, with such costs potentially affecting the election outcome
 \citep{ledyard1984pure,palfrey1983strategic,palfrey1985voter}.} There are many reasons why voting may be costly for an individual: going to the polling station requires effort {and} is associated with opportunity costs, the need to understand {some details about the election process} may discourage some citizens to vote, {or some individuals} may be disappointed from repeatedly being in the minority.

When voting is costly, a balance is struck in equilibrium between such cost and the expected benefit of going to the ballot box. Three stylized facts are then predicted by a major strand of the literature on costly voting \citep[see e.g.][]{borgers,mandatory,taylor-yildirim-2,taylor-yildirim-1}.\footnote{Most papers in the costly-voting literature analyze private-value settings. We refer to \cite{ghosal2009costly} for a setting where preferences have both private- and public-value components.} First, if at all, citizens vote for their preferred alternative, and hence no strategic voting occurs.\footnote{This feature obtains even if there are more than two alternatives \citep{polborn} and has empirical support  in settings where voting is voluntary \citep{bhattacharya2014compulsory}.} Second, regardless of the distribution of preferences within the entire citizenry, both alternatives are expected to win with the same probability.\footnote{While sometimes less stark, the \textit{underdog effect}---according to which supporters of the minority alternative turn out in relative terms more than supporters of the majority alternative---is featured by most models. Exceptions include the case where the cost is much smaller for the members of the majority and the case where there is ambiguity about the true preferences of the electorate \citep{taylor-yildirim-2}.} Third, absolute aggregate turnout is bounded from above, regardless of the size of the electorate. This implies that relative turnout decreases with the size of the electorate. The voting procedure for which {these} three properties are derived is the standard and widely-applied one-round voting, which is modeled by having all citizens (simultaneously) decide whether or not to go to the ballot box, and in the former case which alternative to vote for. 

The goal of the present paper is to show that within a costly-voting set-up, it is possible to devise a different voting procedure which is superior to the standard one-round voting in the following sense: the final decision will match the preference of the majority of the population with a probability {arbitrarily} close to one, and, also in expectation, participation costs will be similar to the participation costs in the {(voluntary)} one-round voting procedure. At the same time, the suggested procedure is also superior to compulsory one-round voting, in which case the alternative preferred by the majority is also chosen with high probability, but at a much higher cost. As a matter of fact, we will argue that there is a sense in which the {proposed} voting procedure can be seen as the right mix of voluntary and compulsory one-round voting schemes. 

Instead of having the entire citizenry vote at the same time, we suggest the following two-round procedure, which we call \textit{Assessment Voting} (AV in short):\footnote{For a verbal description, see~\cite{original}.}
\begin{enumerate}
	\item     A number of citizens are randomly selected from the entire population, {all of whom} constitute the \textit{Assessment Group (AG)}.
	\item    All members of  AG (simultaneously) cast a vote for one of the alternatives at hand or abstain.
	\item     The number of votes in favor of either alternative obtained in the first round is made public.
	\item     All citizens who do \textit{not} belong to AG (simultaneously) decide whether to abstain or to vote for either alternative{, and thus} the second voting round takes place.
	\item     The alternative with the most votes in the two rounds combined is implemented. Ties are broken by a fair toss coin.
	
\end{enumerate}


Because AV serves the purpose of choosing one of two alternatives and is compatible with basic democratic principles---every citizen is granted one vote---, it could be used both in representative and direct democracies for any voting by the entire citizenry, such as the referendum on Brexit.\footnote{In Section~\ref{S:extensions}, we show that Assessment Voting works similarly when there are more than two alternatives.} To be implementable in democratic environments, however, the members of AG would have to be selected truly randomly and this group be large enough to make the results of this first round representative for the entire electorate from an ex post perspective---not merely ex ante. Moreover, both features should be common knowledge.\footnote{Given the {unequal} power that members of both voting groups have, distrust would emerge if citizens were to anticipate that either the identity of such groups may be manipulated or the selected group may not be representative enough (say, because the group is too small).} By building on a sequential voting procedure, and due to a few other features that we will discuss below, AV will be able to equalize the extent of the externalities that voting in favor of either alternative creates on its supporters.\footnote{A strand of literature on costly-voting has analyzed sequential procedures, {with the focus} on information aggregation \citep[see e.g.][]{battaglini2005sequential,optimal_voting}.} It is actually known from the literature on costly voting that externalities generated by casting a ballot typically yield a welfare-inefficient level of turnout. This phenomenon is at the core of the drawbacks of voluntary one-round voting procedures, and is resolved by AV, which induces an endogenous level of turnout that yields socially desirable outcomes.


To assess the properties of AV, we consider a model of a society that needs to choose one of two alternatives, say $A$ and $B$. Each citizen's preference is private information and is independently drawn from a given common distribution. We assume that ex ante, it is more likely that a citizen prefers $A$ to $B$ than $B$ to $A$. This means that $A$ is the desirable alternative from an ex-ante utilitarian perspective. For each citizen, there is also a cost $c>0$ of going to the ballot box. Such participation costs are private, but they may also be considered in the societal calculus: from {a} utilitarian perspective, right decisions should be met at the lowest possible average cost of participating in the voting procedure, given standard democratic constraints such as the right of every citizen to vote. Because we consider large societies, we assume that the number of citizens follows a Poisson distribution---and hence our political game is a Poisson game \citep{poisson,largepoisson}.\footnote{The costly-voting literature has shown that Poisson games characterize the limit {scenario where} the number of citizens goes to infinity \citep[see e.g.][]{taylor-yildirim-1}. Accordingly, considering a Poisson game does not drive our results, but it simplifies the analysis greatly. Recent papers that study Poisson models are \cite{campbell1999large}, \cite{hughes}, or \cite{polborn}. Also recently,  \cite{poissonGames} have analyzed the structure and number of Nash equilibria in Poisson games under different voting schemes.} 

The characterization of the equilibria of our sequential game is in general a complex task, even if we focus on the customary type-symmetric, totally-mixed strategy equilibria. The main reason for this complexity stems from the fact that for a two-round voting procedure such as AV, the strategies of second-round citizens need to take the outcome of the first round into account. In turn, the first-round voters face two sources of uncertainty: within-round uncertainty (how will the other members of AG vote, if at all), and across-round uncertainty (how will the second-round citizens vote in response to the outcome of the first round and to the predicted votes of all other members of the second-round group, if at all). 

In many cases, the above features yield a multiplicity of equilibria.\footnote{In some countries, unofficial voting polls are revealed before all polling stations are closed. This is the case in Spain, where the official turnout rate is also revealed in the course of election day, from which certain information about the development of the voting outcome can, in principle, be extracted. Our analysis reveals that, even if we leave aside the strategic incentives of choosing the moment for going to the ballot box, the purely positive analysis of sequential costly voting is a very difficult task.} Nevertheless, we shall prove that if in the first round voting is compulsory---or is incentivated through subsidies---and the size of the first group is sufficiently large, only one equilibrium of the subgame starting after the publication of the first-round vote count survives: no citizen will cast a vote in the second voting round. This implies that the outcome (i.e., the alternative chosen and the costs of voting incurred by all citizens) is fully determined in the first round. While this is admittedly a very strong prediction, it is reasonable to expect that the main mechanisms underlying this prediction will also operate in real-world environments, thereby giving considerable power of decision to members of the first voting group.\footnote{It is easy to verify that as long as the share of non-strategic voters---say those who always vote regardless of any other consideration---is not sufficiently large (relative to the size of AG), the outcome will still exhibit the same properties as in our baseline model, and hence the welfare conclusions will be very similar.} If, as already highlighted, the composition of this group is representative enough of the entire citizenry, socially optimal alternatives will be chosen at a low societal cost, without the need to deprive citizens of their right to vote: the low level of turnout in the second round will simply arise as the result of a cost-benefit analysis made by the citizens participating in this voting round, all of whom will be aware of the result in the first round.


As a consequence, in the case of AV, the two components of welfare will also be determined entirely by the outcome of the first voting round. On the one hand, the alternative will be resolved by the (random) composition of such group, and hence the probability that the socially desirable alternative $A$ will be chosen goes to one as the size of such a group increases. What is more, the expected value of the distribution of the first-round vote count difference in favor of alternative $A$ (i.e. votes for $A$ minus votes for $B$) will also increase with the size of AV. This, in turn, will make it more complicated for any fixed group of B-supporters to change the final outcome in the second round, thereby reducing the individual incentives for each of them to go to the ballot box in the expectation that the negative result from the first round will be overcome. On the other hand, there will be no other costs associated with voting except the costs (or the subsidies) that are necessary to make all members of AG participate in the election process. It turns out that if the citizenry is large enough, it is possible to set the size of AV such that alternative $A$ is chosen with high probability and the voting costs remain moderately low. This follows from the fact that the vote count threshold that discourages participation in the second round voting does not change as the size of the entire electorate increases. 

The remainder of the paper is organized as follows: In Section~\ref{S:model} the model is introduced. In Section~\ref{eq:analysis} we analyze the voting equilibria under AV. In Section~\ref{S:social_welfare} we explore if AV improves welfare compared to one-round voting, whether it is compulsory or voluntary. In Section~\ref{S:extensions} we analyze some extensions of our baseline model---see also Appendix~B. Section~\ref{S:conclusion} concludes. The proofs of the main body of the paper are in Appendix~A.

\section{Model}
\label{S:model}

\subsection{Set-up}

We consider a country---or, more generally, a jurisdiction---whose citizens have a right to vote for one of two alternatives (or candidates), say $A$ and $B$. Citizens are indexed by $i$ or $j$.  There is a number $p$, with $1/2<p<1$, such that citizen $i$'s preferred alternative is $A$ with probability $p=:p_A$ and $B$ with probability $1-p=:p_B$. While individual preferences are private information, their prior distribution---i.e., the value of $p$---is common knowledge. If citizen $i$'s preferred alternative is chosen, he derives utility $1$, while he derives utility $0$ if the other alternative is chosen. This normalization is standard and does no affect the results. On occasion, we may also say that citizen $i$'s type is $t_i=A$ ($t_i=B$) when his preferred alternative is $A$ ($B$). Additionally, if $i$ exercises his right to vote, he incurs a cost $c$, which is additively subtracted from his utility. We consider that\footnote{If $c>1/2$, no citizen has incentives to vote at all. Assuming that $c$ is common to all voters with the same preferences is not a critical assumption, since we consider large populations, in which case the incentives to vote are very small for those citizens with cost higher than the lowest one  \citep[see e.g.][]{taylor-yildirim-2,taylor-yildirim-1}. Assuming that $c$ is common across types of citizens will allow us to focus on the differential effect of AV with respect to standard one-round voting procedures. Similar results would nonetheless obtain in the case where the two types of citizens {incurred} different costs of voting.}
\begin{equation}
\label{condition_costs}
0<c<1/2.
\end{equation}
We summarize the citizen utility profile in Table~\ref{table:utilities}.


\begin{table}[!h]
\begin{center}
  \begin{tabular}{ c | c | c}
    \hline
    & \small{$i$'s preferred alternative is chosen} & \small{$i$'s preferred alternative is \textit{not} chosen}  \\ \hline
    \small{$i$ votes} & $1-c$ & $-c$ \\ \hline
    \small{$i$ does \textit{not} vote} & $1$ & $0$ \\
    \hline
  \end{tabular}
\end{center}
\caption{Voter Utilities.}
\label{table:utilities}
\end{table}

\subsection{A new two-round voting}

Under \textit{Assessment Voting} (AV), there are two voting rounds. In the first round, a small number of citizens are chosen by fair randomization to participate, all of whom constitute the so-called \textit{Assessment Group (AG)}. That is, each citizen has the same probability to be a member of AG. We let $N_1$, a positive integer, denote the size of  AG. All members of AG (simultaneously) decide whether to exercise their right to vote or not, and if so, which alternative to vote for. We assume that all members of AG are given a subsidy equal to $c$, their cost of voting, so that voting in the first round becomes costless for them. Whether members of the first round exercise their right to vote or not, they  cannot vote in the second round. In the latter round, only citizens who are \textit{not} members of AG have a right to vote. Before the voting in the second round takes place, the number of votes that each alternative received in the first voting round is disclosed and becomes common knowledge. Henceforth, we let $d$ denote the vote difference between alternatives $A$ and $B$ in the first round. In particular, if $d>0$, $A$ received $d$ more votes than $B$ from the members of AG. The alternative that receives more votes within the two voting rounds combined is implemented, with ties being broken by fair randomization.

We assume that the \textit{total} number of citizens is $N=N_1+N_2$, where $N_2$ follows a Poisson distribution with parameter $n_2$, {with} $n_2$ being a positive real number.\footnote{This assumption is made for convenience, but it does not affect our results qualitatively. Alternatively, we could assume that the total number of citizens $N$ follows a Poisson distribution with parameter $n$. In that case, for a fixed $N_1$, the probability that there are not enough citizens to make up for the $N_1$ members of AG converges to zero, as we {increase} $n$.} Then, we let $n=N_1+n_2$ denote the expected number of citizens. Following \cite{largepoisson}, the number of citizens of type $t$ in the second round, with $t\in \{A,B\}$, follows a Poisson distribution with parameter $n_2\cdot p_t$. The properties of the Poisson distribution ensure that from the perspective of a voter of type $t$, the number of voters of his same type also follows a Poisson distribution with parameter $n_2\cdot p_t$. This will simplify the analysis greatly. Finally, we denote by $\Omega_1$ and $\Omega_2$ the set of citizens of the first and second voting round, respectively.



\subsection{Equilibrium concept and information}

We study the existence and multiplicity of type-symmetric perfect Nash equilibria in our voting game. By \textit{type-symmetric} we mean that within each round, all citizens of the same type use the same strategy. Moreover, we assume that if they \textit{do} turn out, they vote \textit{sincerely}, i.e., we assume that they either vote in favor of their preferred alternative or abstain. In the second round, sincere voting arises endogenously as in one-round voting procedures already analyzed in the literature \citep[see e.g.][]{taylor-yildirim-1,polborn}. This follows from the fact that once the results of the first round become common knowledge, voting for an alternative that is not one's preferred is a weakly-dominated strategy for any citizen. As for the first round, although we impose sincere voting as an assumption of our model, it will turn out to be compatible with equilibrium behavior.\footnote{This assumption is further discussed {in Section \ref{subsec:first_round}}.} In combination with the subsidies given to members of AG, this means that the first-round outcome, namely $d$, follows mechanically from the size of AG and the value of $p$. The reason is that every member of AG will vote, and he will do it for the alternative he prefers. Accordingly, let citizen $i$ be a member of AG and consider the following random variable:
\begin{align}
\mathcal{X}_i=\begin{cases}
+1 & \mbox{ if }t_i=A ,\\
-1 & \mbox{ if }t_i=B.
\end{cases}
=
\begin{cases}
+1 & \mbox{ with probability } p_A,\\
-1 & \mbox{ with probability } p_B.
\end{cases}
\label{random_variable}
\end{align}
Then, $d$ is the outcome of the random variable $D$ defined by
\begin{align}
\label{def:sum_votes_first_round}
D := \sum_{i \in \Omega_1} \mathcal{X}_i.
\end{align}
As far as the citizens' strategic choices are concerned, we can thus focus on the subgame starting after the first voting round and after the value of $d$ has been made public, which we denote by $\mathcal{G}^2(d)$. For simplicity, we assume that the citizens who vote in the second round can only condition their vote on their type and the observed value of $d$, since nothing else is payoff-relevant. Accordingly, a strategy for citizen $i$ is a mapping
\begin{align*}
\alpha_i : \{A,B\}\times \{-N_1,\ldots,0,\ldots,N_1\} \rightarrow [0,1].
\end{align*}
That is, $\alpha_i(t, d)$ indicates the probability of citizen $i$ voting for his preferred alternative if he is of type $t$ and the vote difference between the two alternatives in the first round is $d$. As is standard, we assume that there are mappings
\begin{align*}
\alpha_A : \{-N_1,\ldots,0,\ldots,N_1\} \rightarrow [0,1]  \quad \mbox{ and } \quad \alpha_B : \{-N_1,\ldots,0,\ldots,N_1\} \rightarrow [0,1] 
\end{align*}
such that $\alpha_i(A,d)=\alpha_A(d)$ if $t_i=A$ and $\alpha_i(B,d)=\alpha_B(d)$ if $t_i=B$. That is, the probability that citizens of the same type will turn out are the same. A strategy profile is denoted by $\alpha=(\alpha_A,\alpha_B)$. Finally, we define $d_A=d$ and $d_B=-d$.

\section{Analysis of Assessment Voting}\label{eq:analysis}

We start by analyzing the second round of AV, which is described by $\mathcal{G}^2(d)$, and then focus on the analysis of the entire voting procedure.

\subsection{Second voting round}\label{subsec:second_round}

In the second voting round of AV, citizen $i$'s vote will make a difference in the final outcome only if the votes---together with the abstentions---of the remaining citizens that have a right to cast a vote in this round are such that:
\begin{itemize}
	\item in the second round, $i$'s preferred alternative obtains $d_{t_i}+1$ votes less than the other alternative, or 
	\item in the second round, $i$'s preferred alternative obtains $d_{t_i}$ votes less than the other alternative.
\end{itemize}
In the first case, $i$'s vote in favor of his preferred alternative $t_i$ will turn a defeat of $t_i$ into a tie, while in the second case, $i$'s vote in favor of $t_i$ will turn a tie into a win of $t_i$. In both cases, (expected) utility increases by $1/2$ if citizen $i$ turns out and votes in favor of his preferred alternative.
 


In the following, we investigate the totally-mixed equilibria of $\mathcal{G}^2(d)$, i.e., we assume that $0<\alpha_i(d)<1$ for $i\in \{A,B\}$. This type of equilibria is central in the costly-voting literature \citep[see e.g.][]{taylor-yildirim-1,polborn}. It will come in handy to use $x_A:=n_2 p_A \alpha_A$ and  $x_B:=n_2 p_B \alpha_B$ to denote the expected number of votes for each alternative given strategy profile $\alpha$. Note that it is equivalent to determine the pair $(\alpha_A(d),\alpha_B(d))$ and to determine the pair $(x_A,x_B)=(x_A(d),x_B(d))$. We now derive the conditions that make both type of citizens indifferent between abstaining and voting in favor of their preferred alternative, thereby incurring cost $c$. First, we assume that $d=0$. Then, we obtain the following two equations:
\begin{align}\label{type_A_m_zero}
	c & = \frac{1}{2}\sum_{k = 0}^{\infty} \frac{x_A^k}{e^{x_A}k!}\frac{x_B^{k}}{e^{x_B}k!} + \frac{1}{2} \sum_{k = 0}^{\infty} \frac{x_A^k}{e^{x_A}k!}\frac{x_B^{k+1}}{e^{x_B}(k+1)!},
 \\ \label{type_B_m_zero}
c & = \frac{1}{2}\sum_{k = 0}^{\infty} \frac{x_A^k}{e^{x_A}k!}\frac{x_B^{k}}{e^{x_B}k!} + \frac{1}{2} \sum_{k = 0}^{\infty} \frac{x_A^{k+1}}{e^{x_A}(k+1)!}\frac{x_B^{k}}{e^{x_B}k!}.
\end{align}
The first equation corresponds to the indifference condition for any voter $i$ of type $t_i=A$, while the second equation is the indifference condition for any voter $i$ of type $t_i=B$. Mathematically, the case where $d=0$ corresponds to the case where there is only one round of simultaneous, voluntary voting.\footnote{When voting is compulsory, the analysis is almost trivial: all citizens vote for their preferred alternative and they incur the cost of voting.} By simple algebraic manipulations, we obtain $x_A=x_B=x$, where
\begin{equation}
\label{eq:simultaneous_voting}
x \cdot  \sum_{k = 0}^{\infty} \frac{x^{2k}}{k!(k+1)!} =2 c e^{2x} - \frac{1}{2}\sum_{k = 0}^{\infty} \frac{x^{2k}}{k!k!}.
\end{equation}
The above equation has a unique solution in the unknown $x$ \citep[see][]{polborn}. Second, we assume that $d\geq 1$.\footnote{The case $d\leq -1$ {can be proven} analogously.} Then, we obtain the following system of equations:
\begin{align}\label{type_A}
c & = \frac{1}{2}\sum_{k = 0}^{\infty} \frac{x_A^k}{e^{x_A}k!}\frac{x_B^{k+d}}{e^{x_B}(k+d)!} + \frac{1}{2} \sum_{k = 0}^{\infty} \frac{x_A^k}{e^{x_A}k!}\frac{x_B^{k+d+1}}{e^{x_B}(k+d+1)!},
\\
\label{type_B}
c & = \frac{1}{2}\sum_{k = 0}^{\infty} \frac{x_A^k}{e^{x_A}k!}\frac{x_B^{k+d}}{e^{x_B}(k+d)!} + \frac{1}{2} \sum_{k = 0}^{\infty} \frac{x_A^k}{e^{x_A}k!}\frac{x_B^{k+d-1}}{e^{x_B}(k+d-1)!}.
\end{align}
The following result, which is {shown} in Appendix~A, {demonstrates} that the above system of equations is incompatible:
\begin{proposition}
	\label{main_result}
There exists $d^*(c)$ such that for all $d \geq d^*(c)$, the system of equations defined by (\ref{type_A}) and (\ref{type_B}) has no solution. Moreover, $d^*(c)$ (weakly) increases as $c$ decreases.
\end{proposition}	

The negative result identified by Proposition~\ref{main_result} does not follow from the fact that the two equations of the system are incompatible, but from the fact that each of them cannot separately hold for values of $d$ that are large enough. To show this property more clearly, we assume now that $d\geq 2$ and focus on equilibria of the following type: citizens of type $A$ vote with probability zero, while citizens of type $B$ randomize. Note that because $d\geq 2$, an equilibrium where only citizens of type $A$ vote with positive probability cannot be an equilibrium, as alternative~$A$ will be chosen with certainty in the absence of any further votes.\footnote{If $n_2$ is large enough, it can be easily verified that there cannot be an equilibrium where at least one type of citizens votes with probability one.} Hence, we assume that $\alpha_A=0$ and $0<\alpha_B<1$, and obtain the following two conditions:
\begin{equation}\label{0_probability_A}
2c \geq \frac{x_B^d}{e^{x_B}d!} + \frac{x_B^{d+1}}{e^{x_B}(d+1)!}
\end{equation}
and
\begin{equation}
\label{0_probability_B}
2c = \frac{x_B^d}{e^{x_B}d!} + \frac{x_B^{d-1}}{e^{x_B}(d-1)!}.
\end{equation} 
The first equation guarantees that citizens $i$ of type $t_i=A$ are content with their decision not to vote, while the second equation is the indifference condition for any voter $i$ of type $t_i=B$. We can prove the following lemma:

\begin{lemma}
	\label{main_lemma}
	There exists a positive integer $d^*(c)$ such that Eq. (\ref{0_probability_B}) does not have a solution for all $d \geq d^*(c)$. 
\end{lemma}

We point out that the threshold $d^*(c)$ of Lemma~\ref{main_lemma} is precisely the threshold used in Proposition~\ref{main_result}, and that we will use the same notation throughout the paper, including the Appendices. Furthermore, it is trivial to note that if $d\geq 2$, there is an equilibrium in which no citizen votes---we call it the \textit{no-show equilibrium}. The combination of Proposition~\ref{main_result} and Lemma~\ref{main_lemma} leads to the following result:



\begin{corollary}\label{corollary:non_existence_result}
If $d\geq d^*(c)$, the only equilibrium of $\mathcal{G}^2(d)$ is the no-show equilibrium. 
\end{corollary}

 According to Corollary~\ref{corollary:non_existence_result}, if the absolute vote difference between the two alternatives  in the first voting round, namely $\vert d \vert$, is large enough,  there are no incentives for any second-round citizen to participate in the second voting round.\footnote{By symmetry, the case $d\leq -d^*(c)$ is analogous and hence the no-show equilibrium is the only equilibrium.} What is more, this property holds regardless of the (expected) size of the second-round voting group. An ensuing question is what the outcome from the second round is when $\vert d \vert $ is moderately low. We obtain the following result:



\begin{lemma}
	\label{first_root_bis}
Given $d> 2$, there is $c^*(d) \in (0,1/2)$ such that for all $c<c^*(d)$, an equilibrium $(0,x_B)$ of $\mathcal{G}^2(d)$ exists.
\end{lemma}

The above lemma complements the result of Corollary~\ref{corollary:non_existence_result}. While $d^*=d^*(c)$ determines the size of AG above which the no-show equilibrium is the only equilibrium, $c^*=c^*(d)$ determines the cost level below which equilibria that are different from the no-show equilibrium exist. It turns out---see the proof of Proposition \ref{main_result}---that for any given $c\in (0,1/2)$, there exist constants $K_1$ and $K_2$, with $K_2<K_1$, such that 
\begin{itemize}
	\item [\textit{(i)}] if $d>\frac{K_1}{c^2}$,  the only equilibrium of $\mathcal{G}^2(d)$ is the no-show equilibrium, and
	\item [\textit{(ii)}] if $d <\frac{K_2}{c^2}$, then $\mathcal{G}^2(d)$ has equilibria different from the no-show equilibrium.
\end{itemize}	
Hence, we can say that both thresholds are (approximately) tight, in the sense that $d^* \sim \frac{1}{(c^*)^2}$. In addition, it can be verified numerically that uniqueness of equilibria of $\mathcal{G}^2(d)$ is not guaranteed within all admissible parameter ranges, even if we only consider equilibria of the type $(0,x_B)$. For instance, multiplicity of equilibria occur if we consider $c=0.2$ and $d=3$. In this case, we have that $(0, y_1)$ and $(0,y_2)$ are equilibria of $\mathcal{G}^2(3)$, where $y_1\approx 3.17$ and $y_2\approx 3.76$ are positive solutions of the equation $0.4e^y=\frac{y^3}{6}+\frac{y^4}{24}$ that additionally satisfy the inequality $0.4e^y>\frac{y^4}{24}+\frac{y^5}{120}$.\footnote{The same holds true if we restrict to equilibria of the type $(x_A,x_B)$. Numerical examples for this other case can be provided upon request.} This example shows that if $\vert d \vert$ is moderately low, we cannot uniquely predict the outcome of the second round of AV, if at all. 



\subsection{First voting round}\label{subsec:first_round}

Corollary~\ref{corollary:non_existence_result} yields a very strong prediction: if $d$ is above a certain threshold, no citizen will vote in the second voting round. It turns out that by making $N_1$, the size of AG, large enough, the probability that $d$ is larger than this threshold converges to one. This is proved in the following result, which characterizes the outcome of Assessment Voting (almost surely). 


\bigskip

\begin{theorem}\label{thm:main_theorem}
For every $\varepsilon>0$, there is $N_1^*=N_1^*(\varepsilon,c,p_A-p_B)$ such that for all $N_1 \geq N_1^*$, the outcome of AV satisfies the following properties  with probability at least $1-\varepsilon$:
\begin{itemize}
	\item All citizens of the first voting round vote for their preferred alternative.
	
	\item No citizen of the second voting round votes.
	
	\item Alternative $A$ is chosen.
\end{itemize}
\end{theorem}

According to the above theorem, if AG is large enough, citizens who have a right to vote in the second round are  all discouraged from going to the ballot box. The logic behind this result hinges on the law of large numbers: because \textit{(i)} alternative~$A$ is more preferred in the society than alternative $B$, \textit{(ii)} members of AG are selected randomly, and \textit{(iii)} voting is subsidized for members of AG, the difference in the first-round vote count for alternative~$A$ with respect to alternative~$B$ increase{s} with the size of $AG$, until the no-show equilibrium is reached (with high probability). {Moreover, this is} the only equilibrium in the second-round voting game (with this same high probability). The following corollary follows from Theorem \ref{thm:main_theorem}, and reveals how the size of AG should vary with respect to the most important parameters of the model:

\begin{corollary}\label{cor:main_theorem}
Let $N_1^*=N_1^*(\varepsilon,c,p_A-p_B)$ as defined in Theorem~\ref{thm:main_theorem}. Then,\nopagebreak
\begin{itemize}
	\item  $N_1^*$ increases if $\varepsilon$ decreases, with $\lim\limits_{\varepsilon\rightarrow 0} N_1^* = \infty$,
	\item  $N_1^*$ increases if $p_A-p_B$ decreases, with $\lim\limits_{p_A-p_B\rightarrow 0} N_1^* = \infty$,
	\item  $N_1^*$ increases if $c$ decreases, with $\lim\limits_{c\rightarrow 0} N_1^* = \infty$.
	
\end{itemize}
\end{corollary}

The behavior of $N_1^*$ with respect to changes in $\varepsilon$ and $p_A-p_B$ is {self-}evident: when either the society is more divided (i.e., lower $p_A-p_B$) or we want to be more certain that the voting outcome will be dictated entirely by AG members (i.e., lower $\varepsilon$), the size of AG needs to be greater. Most remarkably, the above corollary implies that, ceteris paribus, a lower size of AG is obtained when $c$ increases. In particular, incentives for voting in the second round are least significant for $B$-supporters relative to those of $A$-supporters when $c$ is (almost) equal to $1/2$, in which case a smaller vote count difference from the first round suffices to discourage voting. In this case, $N_1^*$ may still be large, depending on the values of $p_A-p_B$ and $\varepsilon$. {Note that} $N_1^*$ must be generally large enough to satisfy two objectives: {on the one hand, second-round citizens' incentives} to vote should disappear; {on the other hand}, alternative $A$ should be chosen with at least probability $1-\varepsilon$. It is also important to stress that although $N_1^*$ gives a sufficient condition with regard to the size of AG for the outcome of AV to be described by Theorem~\ref{thm:main_theorem}, the discussion at the end of Section \ref{subsec:second_round} shows that this required size is (approximately) tight, in the sense that the desired outcome may fail to hold for lower AG size. Finally, Table~\ref{table:examples_numbers} depicts the value of $N_1^*$ for some parameter constellations.\footnote{We stress that the values depicted in Table~\ref{table:examples_numbers} do \textit{not} depend on the total number of citizens in the population. We stress that in actual implementations, $N_1^*$ should be chosen taking also into account that the outcome of the first-voting round should represent the actual population.}

 \begin{table}[!htb]

 	\begin{minipage}{.5\linewidth}
 		
 		 {
 		\centering
 		\small

 		\begin{tabular}{c | c | c |}
 		& $p_A-p_B =0.05$ & $p_A-p_B=0.15$  \\ \hline
 		$\varepsilon=0.1$ & 146,049 & 45,945 \\ \hline
 		$\varepsilon=0.01$ & 152,788 & 47,160 \\\hline
 	\end{tabular}\\ \vspace{0.1cm}	
  $c=0.005$ (with $d^*(c)=6,367$)
  
}
 	\end{minipage}%
 	\begin{minipage}{.5\linewidth}
 		
 		 {
 		\centering
 			\small
 				\begin{tabular}{c | c | c |}
 				& $p_A-p_B =0.05$ & $p_A-p_B=0.15$  \\ \hline
 				$\varepsilon=0.1$ & 41,856 & 12,433 \\ \hline
 				$\varepsilon=0.01$ & 45,769 & 13,097 \\\hline
 			\end{tabular}\\ \vspace{0.1cm}	
 		$c=0.01$ (with $d^*(c)=1,592$)
 	}
 	\end{minipage} 
 
 \vspace{0.5cm}
 
  	\begin{minipage}{.5\linewidth}
  		
  		 {
 	\centering
 	\small
 	\begin{tabular}{c | c | c |}
 		& $p_A-p_B =0.05$ & $p_A-p_B=0.15$  \\ \hline
 		$\varepsilon=0.1$ & 3,003 & 455 \\ \hline
 		$\varepsilon=0.01$ & 4,858 & 668 \\\hline
 	\end{tabular}\\ \vspace{0.1cm}	
 	$c=0.1$ (with $d^*(c)=16$)
 	
 }

 \end{minipage}%
 \begin{minipage}{.5\linewidth}
 	
 	 {
 	 	
 	\centering
 	\small
 	\begin{tabular}{c | c | c |}
 		& $p_A-p_B=0.05$ & $p_A-p_B=0.15$  \\ \hline
 		$\varepsilon=0.1$ & 2,476 & 293 \\ \hline
 		$\varepsilon=0.01$ & 4,319 & 498  \\\hline
 	\end{tabular}\\ \vspace{0.1cm}	
 	$c=0.3$ (with $d^*(c)=2$)

}

 \end{minipage} 
  	\caption{The (optimal) size of the Assessment Group (AG).}
		 \label{table:examples_numbers}
 \end{table}

The numbers in Table \ref{table:examples_numbers} reflect the desirable size of AG for some particular situations. We stress that members in this group will exercise their right to vote and that this \textit{hard fact} is very different from participation in pre-election polls, in which case cheap-talk or other strategic behavior may lead to biased revelation of preferences \citep{goeree2007welfare,agranov2012makes}. As a matter of fact, the literature on costly voting predicts that manipulation of polls may have a strong impact on election outcomes, by triggering a level of turnout in equilibrium that does not match the true preferences of the electorate \citep{taylor-yildirim-1,taylor-yildirim-2}. Because AV is based on actual votes and not on reported opinions, such a voting procedure should be more immune to this type of manipulation. The reason is that as a result of the cost-benefit analysis of voting being equalized for all citizens, the outcome of AV matches the prior distribution of preferences in the entire citizenry (with high probability). 

\section{Welfare Analysis} 
\label{S:social_welfare}



Having characterized the (almost certain) equilibrium outcome under AV, a natural question is what would be the welfare consequences of introducing such a voting procedure. Focusing on expected average utilitarian welfare, there are two standard benchmarks, both of which consist in a single round of simultaneous voting: first, voting may be voluntary; second, voting may be compulsory.\footnote{The comparison between voluntary and compulsory one-round voting is {the topic of} \cite{borgers} and \cite{mandatory}.} In both cases, we assume that the total number of citizens follows a Poisson probability distribution of parameter $N_1+n_2$. Under one-round voluntary voting, the analysis in Section~\ref{S:model} shows that welfare amounts to
\begin{equation}
\label{eq:welfare_voluntary_one_round}
W^{vol}:=\frac{1}{2}-\frac{x}{N_1+n_2}\cdot c,
\end{equation}
where  $x$ is the solution to Equation (\ref{eq:simultaneous_voting}). Under one-round compulsory voting, it is easy to verify that welfare amounts to
\begin{equation}
\label{eq:welfare_compulsory_one_round}
W^{com}:=w_{d}^{com}- c,
\end{equation}
where $w_{d}^{com}$ is the expected average welfare obtained from the alternative eventually implemented when the entire population, which has expected size equal to $N_1+n_2$, votes sincerely. It is easy to verify that $w_d^{com} = p_A  \cdot (1-z^{com}_d(N))$, with $\lim_{N\rightarrow \infty} z^{com}_d(N)=0$. Moreover, according to Theorem~\ref{thm:main_theorem}, there is $N_1^{*}=N_1^{*}(\varepsilon,c,p_A-p_B)$ such that for all $N_1 \geq N_1^{*}$, with probability at least $1-\varepsilon$ we have 
\begin{equation}
\label{eq:welfare_assessment_voting}
W^{AV}:=p_A-\frac{N_1}{N}\cdot c,
\end{equation}
where $p_A$ coincides with the expected welfare obtained from the alternative being eventually implemented when all members of AG (which has a certain size $N_1$) vote sincerely. Finally, a lower bound for (expected) welfare is
\begin{equation}
\label{eq:welfare_lower_bound}
\underline{W}:=0- 1 \cdot c = -c.
\end{equation}
We will consider $\underline{W}$ to estimate welfare when our analysis does not yield clear-cut predictions as to the outcome of AV. Our main result regarding welfare is the following:
\begin{theorem}\label{thm:welfare}
There are $N_1^{**}(c)$ and $n_2^{*}(c)$ such that if $N_1\geq N_1^{**}(c)$ and $n_2\geq n_2^*(c)$, we have
\begin{equation*}
W^{AV} > \max \{W^{vol},W^{com}\}.
\end{equation*}
\end{theorem} 

Naturally, at the constitutional level, neither $N_1^{**}(c)$ nor $n_2^{*}(c)$ could depend on the particular instances of referenda that would take place, which would be characterized by different parameters, particularly by different values of $p_A-p_B$. Theorem~\ref{thm:welfare} nonetheless indicates that in large societies, AV will perform better  on average than standard one-round voting, whether voting in the latter is voluntary or compulsory. In comparison with voluntary one-round voting, participation costs in AV will be of a similar extent, but decisions will represent the population preferences much more accurately. In comparison with compulsory one-round voting, decisions will represent the population preferences equally well, but participation costs will be much lower in AV. Hence, AV simultaneously exhibits the most desirable properties of voluntary and compulsory one-round voting, and it can thus be seen as an appropriate mixture of both approaches. It should also be emphasized that AV exhibits the main desirable features of {democratic mechanisms}, including the fact that every citizen has a right to vote. This property adds to the appeal of this new voting mechanism. What is more, alternatives that find little support in the citizenry are bound to be defeated in equilibrium {when AV is used}. In direct democracies, this fact may reduce the incentives to initiate popular voting on issues which are only supported by a small minority.

\section{Extensions}\label{S:extensions}


The baseline model {analyzed thus far} can be extended in at least two sensible ways. First, one may ask whether the prediction that citizens of the second round will (almost) never vote if the size of AG is large enough hinges on our equilibrium concept, and very particularly on the assumption that citizens of the same type \textit{all} use the same strategy. Second, although we have introduced AV for binary decisions, one may wonder about whether the performance of this voting procedure (with respect to voluntary and compulsory one-round voting) is maintained when there exist three or more alternatives. It turns out that the answer to the first question is negative and the answer to the second question is positive. A comprehensive analysis of both extensions can be found in Appendix~B.  

On the one hand, it is important to {investigate whether} the assumption that all citizens with the same preferences use the same strategy does drive our main (negative) result---namely that no equilibria differing from the no-show equilibrium exist if the vote count difference in the first round is large enough (in absolute terms){---or not. Demonstrating} that this result holds even if we consider different (sub)types for citizens who have the same preferences, with each subtype using a different strategy, greatly adds to the robustness of our prediction regarding the outcome of AV. The proof---see Appendix~B---is based on properties of the Poisson distribution and the multinomial theorem. On the other hand, a setting with three or more alternatives {allows} us to extend the application of AV from binary decisions (i.e., referenda) to other decisions, say elections for executive offices where several candidates compete. As mentioned in the Introduction, the case of multiple alternatives has been recently studied by \cite{polborn}. As in their paper, we show that sincere voting---i.e., voting in favor of the preferred alternative---is consistent with equilibrium behavior, although there may exist other equilibria. 

Finally, we note that the robust result that identifies the conditions for which the {no-show equilibrium is the} only equilibrium that will survive, provided that the vote count difference is large enough, has been derived within the framework of AV. Nevertheless, our analysis of game $\mathcal{G}^2(d)$ {could} be applied to one-round voluntary voting procedures where one of the two alternatives, say alternative $A$, is the status quo, and the other alternative, say alternative $B$, has to reach a qualified majority in order to be implemented.\footnote{A qualified majority imposes no direct requirement on the \textit{absolute} difference needed between the number of votes in favor of the alternative and the number of votes in favor of the status quo, but on the \textit{relative} difference. However, under the assumption that there will always be ``by default'' some share of partisan citizens voting for the status quo, it \text{does} also impose a requirement on this absolute difference. {This allows the applicability of our analysis to this second set-up.}} In that case, not only {would $d$ typically} be positive, but it {would} typically be very large if the electorate {were} large itself. Our (negative) result {would then imply} that in settings where voting is costly and the alternatives at hand only have a private value component, qualified majorities {might} effectively protect the status quo in general, regardless of the support such an alternative gathers within the population. We stress that under the standard (non-qualified) majority rule, the literature on costly voting predicts the opposite: the status quo and the alternative proposal will be implemented with the same probability, regardless of the support that either alternative gathers within the population.

\section{Conclusion}\label{S:conclusion}

Most democracies, representative or direct, have faced important challenges recently. Some of these challenges were due to inefficiencies of the decision mechanisms. In Switzerland, for instance, the 100,000-signature threshold for popular initiatives {is easy to attain}, paving the way for a misuse of popular votes as political mobilization devices. With more popular votes in the form of referenda, organization and opportunity costs represent an increasingly important factor to be taken into account.\footnote{For instance, information needs to be distributed to all citizens, with such costs being ultimately financed by taxes.  Moreover, campaign absorbs time from other governmental activities.} {Facing} more referenda also demands more effort from the citizens themselves, especially {in the case of} decisions about which citizens are ex ante poorly informed. Over and above these concerns, some of the referenda that took place in the EU in the last decade have also demonstrated that sometimes decisions end up being strongly dependent on {turnout}, a feature compatible with {the outcome volatility} predicted by the literature on costly voting. 



In this paper, we have advocated a new voting procedure{, which we have called Assessment Voting,} that fulfills the most standard democratic requirements (e.g. one person, one vote) and may be a partial remedy to the problems just described. {The reason is that} Assessment Voting lowers {the} costs of popular votes and ensures (approximately) that the majority/minority relation in the citizenry is reflected in the voting outcome. Although our main analysis has focussed on binary decisions (i.e., referenda), we show in the Appendix that the main mechanisms are at work with a voting on three or more alternatives. {Because Assessment Voting} may yield more informed and less costly collective decisions, it could be tested in democracies on an experimental basis.

Our analysis could be extended in various. First, we could consider circumstances {where} citizens may have only partial knowledge about their own preferences. A sequential voting procedure such as Assessment Voting opens up the possibility for new forms of information transmission from voters of the first voting group to voters of the second voting group. Second, in anticipation of the use of Assessment Voting, proposal-making may change. For instance, proposals that have no chance under Assessment Voting (but {\it do} have one in single-round voting) may not be made anymore. These issues are left for future research.




\bibliographystyle{apalike}
\bibliography{sample}

\pagebreak

\small

\appendix
\section*{Appendix A}

In this Appendix we prove Lemma~\ref{main_lemma}, Proposition~\ref{main_result}, Lemma~\ref{first_root_bis}, Theorem \ref{thm:main_theorem}, Theorem~\ref{thm:welfare}, and Proposition~\ref{three_alternatives}.

\begin{proof}[Proof of Lemma \ref{main_lemma}:]
	
	The goal of the proof is to show that there does not exist a non-negative solution in $y$ for the following equation if $d$ is sufficiently large:
	\begin{equation}
	\label{0_probability_B_proof}
	2c = \frac{y^d}{e^{y}d!} + \frac{y^{d-1}}{e^{y}(d-1)!}.
	\end{equation} 
	We start by noting that the right-hand side of Eq. (\ref{0_probability_B_proof}) is equal to $0$ for $y=0$ and tends to $0$ as $y$ tends to $\infty$. Therefore, proving that Eq. (\ref{0_probability_B_proof}) does not have a  non-negative  solution is equivalent to proving that for all $y\in \mathbb{R}_+$, the left-hand side of Eq. (\ref{0_probability_B_proof}) is strictly larger than the right-hand side.\footnote{We let $\mathbb{R}_+$ denote the set of non-negative real numbers.} To that end, we prove two auxiliary results. First, for a given $d\geq1$, we define
	\begin{align}
	\label{function:auxiliary}
	f_d(y) := ce^y - \frac{y^d}{d!}.
	\end{align}
	We claim that
	\begin{align}
	\label{claim:auxiliary}
	f_d(y)>0 \mbox{ for all }y\in\mathbb{R}_+ \Rightarrow 	f_{d+1}(y)>0 \mbox{ for all }y\in\mathbb{R}_+.
	\end{align}
	For the proof of the claim, assume that the left-hand side of (\ref{claim:auxiliary}) is true. Then, 
	\begin{equation}
	\label{eq:within_proof}
	\frac{\partial f_{d+1}(y)}{\partial y} = f_{d}(y)>0.
	\end{equation}
	That is, $f_{d+1}(y)$ is increasing in $y\in \mathbb{R}_+$. Since $f_{d+1}(0)=c>0$, it follows immediately that the claim in (\ref{claim:auxiliary}) is correct. 	Second, for a given $d\geq2$, define
	\begin{align}
	\label{function:auxiliary_bis}
	g_d(y) := \frac{f_d(y)}{e^y}=c - \frac{y^d}{e^y d!} 
	\end{align}
	and note that
	\begin{align}
	\label{function:auxiliary_bis_bis}
	g_d(y)>0 \Leftrightarrow f_d(y)>0. 
	\end{align}
	Consider now the following claim, which we will also prove:
		\begin{align}
	\label{claim:auxiliary_bis}
	g_d^*(y)>0 \mbox{ for all }y\in\mathbb{R}_+ \mbox{ for some } d^*:=d^*(c)\geq 1.
	\end{align}
	By straightforward calculations, 
	\begin{equation*}
		\frac{\partial g_{d}(y)}{\partial y} = -\frac{y^{d-1}(d-y)}{e^y d! }.
	\end{equation*}
	It then follows that $y^*=d$ is the (global) minimum of $g_{d} (y)$ in $\mathbb{R}_+$, since
		\begin{equation*}
	\frac{\partial g_{d}(y)}{\partial y} \vert_{y=d} = 0
	\end{equation*}
	and $\frac{\partial g_{d}(y)}{\partial y}$ is  negative for all $y<d$ and positive for all $y>d$. 
	We accordingly obtain that, for all $y\in \mathbb{R}_+$,
	\begin{align*}
	g_d (y) \geq g_d (d) = c- \frac{d^d}{e^d d!} \geq c- \frac{1}{\sqrt{2 \pi d} e^{\frac{1}{12d}}},
	\end{align*}
	where the last inequality holds by Stirling's inequality. Hence, a sufficient condition for the claim of (\ref{claim:auxiliary_bis}) to hold is that
	\begin{align*}
	c  > \frac{1}{\sqrt{2 \pi d} e^{\frac{1}{12d}}}.
	\end{align*}
	It is straightforward to verify that the righ-hand side of the above inequality is a decreasing function of $d$, provided that $d\geq 1$, and that, moreover, one that converges to zero as $d$ goes to infinity. Accordingly, we let $d^*(c)$ be (uniquely) defined as the smallest positive integer larger than one that satisfies
	\begin{align}
	\label{claim:auxaux}
	c  > \frac{1}{\sqrt{2 \pi d^*(c)} e^{\frac{1}{12d^*(c)}}}.
	\end{align}
	Note that, in particular,
	 \begin{equation}
	 \label{eq:convergence}
	 d^*(c)={\Omega\left(\frac{1}{c^2}\right)}.
	 \end{equation}
	 All in all, we have demonstrated the claim of Eq. (\ref{claim:auxiliary_bis}). 	Finally, let $d\geq d^*(c)$. Then, for all $y\in \mathbb{R}_+$,
	\begin{align*}
	2c-\left(\frac{y^d}{e^{y}d!} + \frac{y^{d-1}}{e^{y}(d-1)!}\right) = g_{d}(y)+g_{d-1}(y)>0,
	\end{align*}
	where the strict inequality holds by the Claims of (\ref{claim:auxiliary}) and (\ref{claim:auxiliary_bis}). This completes the proof of the lemma.

\end{proof}
	

	
	
	

\bigskip

\begin{proof} [Proof of Proposition \ref{main_result}:]
	
	The goal of the proof is to show that the following system of equations in $(x,y)$ does not have a solution with non-negative components if $d$ is sufficiently large:
	\begin{align}\label{type_A_proof}
	2c & =\sum_{k = 0}^{\infty} \frac{x^k}{e^{x}k!} \left(\frac{y^{k+d}}{e^{y}(k+d)!} +  \frac{y^{k+d+1}}{e^{y}(k+d+1)!}\right),
	\\
	\label{type_B_proof}
	2c & = \sum_{k = 0}^{\infty} \frac{x^k}{e^{x}k!}\left( \frac{y^{k+d}}{e^{y}(k+d)!} +  \frac{y^{k+d-1}}{e^{y}(k+d-1)!} \right).
	\end{align}
	The system of equations is obtained from~(\ref{type_A}) and~(\ref{type_B}) by some algebraic manipulations and by setting $x_A=x$ and $x_B=y$. From the proof of Lemma \ref{main_lemma}, there is a positive integer $d^*=d^*(c)$ such that, for all $d\geq d^*-1$, $k\geq 0$ and $y\in \mathbb{R}_+$,
	\begin{equation}
	\label{eq:aux_1}
	\frac{y^{k+d}}{e^{y}(k+d)!} + \frac{y^{k+d+1}}{e^{y}(k+d+1)!}<2c.
	\end{equation}
	Moreover, it is know from the properties of the Poisson probability distribution that
	\begin{equation}
	\label{eq:aux_2}
	\sum_{k = 0}^{\infty}\frac{x^k}{e^xk!} = 1.
	\end{equation}
	Accordingly,
	\begin{align*}
	\sum_{k = 0}^{\infty} \frac{x^k}{e^{x}k!} \left(\frac{y^{k+d}}{e^{y}(k+d)!} +  \frac{y^{k+d+1}}{e^{y}(k+d+1)!}\right) <
	\sum_{k = 0}^{\infty} \frac{x^k}{e^{x}k!} 2c = 2c ,
	\end{align*}
	where the first inequality is due to (\ref{eq:aux_1}) and the second inequality is due to (\ref{eq:aux_2}). This completes the proof of the proposition, since (\ref{type_A_proof}) cannot be satisfied for any $(x,y)$ with $x,y\in \mathbb{R_+}$.
	
\end{proof}

	

\bigskip

\begin{proof}[Proof of Lemma \ref{first_root_bis}:]

	Throughout the proof, we have $d>2$ fixed. First, we show that an equilibrium $(0,x_B)$ of $\mathcal{G}^2(d)$ exists if and only if Eq. (\ref{0_probability_B}) has a solution. It suffices to prove sufficiency, i.e., if Eq. (\ref{0_probability_B}) holds for a given $(0,x_B)$, this must be an equilibrium of $\mathcal{G}^2(d)$. Indeed, take the smallest positive root of Eq. (\ref{0_probability_B}), which we denote by $x_B^*$. Additionally, consider
	\begin{align*}
	h_{d+1}(y) = 2ce^y - \frac{y^{d+1}}{(d+1)!} - \frac{y^d}{d!} \quad \mbox{ and }	\quad h_{d}(y) = 2ce^y - \frac{y^{d}}{d!} - \frac{y^{d-1}}{(d-1)!}.
	\end{align*}
	That is, $x_B^*$ is the smallest positive solution $y$ of the equation $	h_{d}(y) =0$. In particular, it must be that
		\begin{equation*}
	h_d(x_B^*)=0
	\end{equation*}
	and, by continuity of $h_{d}$ and the fact that $h_{d}(0)=2c>0$,
	\begin{equation}
	\label{0_probability_B_proof_lemma_2}
h_d(y)\geq 0 \mbox{ for all }y\leq x_B^*.
	\end{equation}
	Next, note that from Eq. (\ref{eq:within_proof}) in Lemma \ref{main_lemma}, it follows that
	\begin{equation}
		\label{0_probability_B_proof_lemma_3}
	\frac{\partial h_{d+1}(y)}{\partial y} =\frac{\partial}{\partial y} \left(f_{d+1}(y)+f_{d}(y)\right) = f_{d}(y)+f_{d-1}(y)=h_d(y),
	\end{equation}
	where $f_{d-1},f_d,f_{d+1}$ were defined in (\ref{function:auxiliary}). Hence, Eqs. (\ref{0_probability_B_proof_lemma_2}) and (\ref{0_probability_B_proof_lemma_3}) imply that
	\begin{equation}
	\label{ineq_lemma}
	\frac{\partial h_{d+1}(y)}{\partial y} \geq 0 \mbox{ if }y\leq x_B^*.
	\end{equation}
	Then, (\ref{ineq_lemma}) implies that 
	\begin{equation*}
	2c - \frac{(x_B^*)^d}{e^{x_B^*}d!} - \frac{(x_B^*)^{d+1}}{e^{x_B^*}(d+1)!}=h_{d+1}(x_B^*) \geq h_{d+1}(0)=2c >0.
	\end{equation*}
	As a consequence, Ineq. (\ref{0_probability_A}) is (strictly) satisfied for $x_B^*$, and hence $(0,x_B^*)$ is an equilibrium of~ $\mathcal{G}^2(d)$.

	Second, we show that there is $c^*(d)>0$ such that an equilibrium of~$\mathcal{G}^2(d)$ of the type $(0,x_B)$ exists for all $c\leq c^*(d)$. By the first part of the proof, it is sufficient to prove that such $c^*(d)$ exists guaranteeing that there is $x_B$ such that $h_d(x_B)=0$, provided that $c\leq c^*(d)$. Indeed, let $c^*:=c^*(d)$ be defined as follows:
	\begin{equation*}
	2c^*=\frac{d^{d-1}}{e^d(d-1)!}+\frac{d^d}{e^dd!}.
	\end{equation*}
	Then, for $0<c\leq c^*$,
	\begin{equation*}
	h_d(0)=2c>0
	\end{equation*}
	and
	\begin{equation*}
	h_d(d)=2ce^d-\frac{d^d}{d!}-\frac{d^{d-1}}{(d-1)!}\leq 2c^*e^d-\frac{d^{d-1}}{(d-1)!}-\frac{d^d}{d!}=0.
	\end{equation*}
	Hence, due to continuity of $h_d$, the equation $h_d(x_B)=0$ must have a solution. {This proves the result of the lemma.}
	
	We conclude the proof with a remark.  If we apply Stirling's formula to $c^*(d)=\frac{d^d}{e^d d!}$, we obtain 
	\begin{equation}
	\label{eq:cond_aux_aux}
	c^*(d)=O\left(\frac{1}{\sqrt{d}}\right).
	\end{equation}		 
	 In combination with~(\ref{eq:convergence}), Condition~(\ref{eq:cond_aux_aux}) implies that $d\geq d^*(c)$, with $d^*(c)=\Omega\left(\frac{1}{c^2}\right)$, is not only sufficient the result in Lemma~\ref{main_lemma} to hold, but it is \textit{also} necessary. In other words, the difference $d^*(c)$ in the vote count $d$ obtained after the first voting round \emph{must} be achieved in this round in order for the no-show equilibrium to be the only equilibrium of $\mathcal{G}^2(d)$, the game representing the second-voting round when the vote count difference is $d$. More specifically, for any given $c\in (0,1/2)$, there exist constants $K_1$ and $K_2$, with $K_2<K_1$, such that the following two statements hold. First, if $d>\frac{K_1}{c^2}$,  the no-show equilibrium  is the only equilibrium of $\mathcal{G}^2(d)$. Second, if $d <\frac{K_2}{c^2}$, then $\mathcal{G}^2(d)$ has equilibria that are different from the no-show equilibrium. The existence of $K_1$ and $K_2$ follows from Conditions~(\ref{eq:convergence})~and~(\ref{eq:cond_aux_aux}).	
	 
\end{proof}
	


\bigskip

\begin{proof}[Proof of Theorem \ref{thm:main_theorem}:]


		As already mentioned in the main body of the paper, we assume that all citizens of AG vote sincerely, i.e., that they vote for their preferred alternative. Below, we show that this assumption is also consistent with equilibrium behavior. Accordingly, the behavior of any such citizen $i$ is described by the random variable $\mathcal{X}_i$---see (\ref{random_variable})---, while the difference in vote count for alternative~$A$ with respect to alternative~$B$ obtained in the first voting round is described by the random variable 
\begin{align*}
D = \sum_{i \in \Omega_1} \mathcal{X}_i,
\end{align*}		
which has been defined in (\ref{def:sum_votes_first_round}), $\Omega_1$ denoting the set of citizens that belong to AG. Because $\mathbb{E}[\mathcal{X}_i] = p_A - p_B$ and $X_i$ are i.i.d., it follows that
\begin{equation*}
E[D] = N_1 \cdot \mathbb{E} [\mathcal{X}_i] = N_1 \cdot (p_A-p_B).
\end{equation*}		
Recall that $d^*=d^*(c)$ has been defined in Proposition \ref{main_result}. This integer guarantees that if $d$, the outcome associated with the random variable $D$, is at least $d^*$, the only equilibrium of game $\mathcal{G}_2^*(d)$ is the no-show equilibrium. In that case, the only votes are cast in the first round, and because $d>0$, alternative~$A$ will be chosen. Now, let
\begin{align}
\label{def:N_1_star_}
N_1^* &= N_1^*(c,\varepsilon,p_A-p_B) : = \left\lceil \frac{d^*}{p_A-p_B}{+ \frac{\ln\frac{2}{\varepsilon}}{(p_A-p_B)^2}} + \frac{\sqrt{2d^*(p_A-p_B)\ln \frac{2}{\varepsilon} + (\ln\frac{2}{\varepsilon})^2 }}{(p_A-p_B)^2} \right\rceil.
\end{align}
Henceforth, we assume that
\begin{equation}
\label{def:N_1_assumption}
N_1 \geq N_1^*(c,\varepsilon, p_A - p_B).
\end{equation}	
Then, we obtain that
\begin{equation}
\label{eq:aux_negative}
d^*-E[D] = d^*- N_1 \cdot (p_A-p_B) \leq d^*- N_1^* \cdot (p_A-p_B) <0,
\end{equation}	
where the first inequality follows from Ineq. (\ref{def:N_1_assumption})  and the second inequality follows from the fact that $N_1^* \geq \frac{d^*}{p_A-p_B} $, as implied by the definition of $N_1^*$ in (\ref{def:N_1_star_}).
Then, the following chain of inequalities also holds:	
\begin{align*}
P\left[D \leq d^*\right] = P[D-E[D]  \leq d^{*}-E[D]]  \leq P\left[\vert D-E[D] \vert \geq E[D]-d^* \right],
\end{align*}
where the last inequality holds due to  (\ref{eq:aux_negative}). 
Moreover, by Hoeffding's inequality (\cite{hoeffding}), 
\begin{align*}
P\left[\vert D-E[D] \vert \geq E[D]-d^* \right] & \leq 2 exp\bigg(-\frac{(E[D]-d^* )^2}{2N_1}\bigg) \\
&  = 2 exp\bigg(-\frac{(N_1(p_A-p_B)-d^* )^2}{2N_1}\bigg)  \leq \varepsilon,
\end{align*}
where the last inequality holds by~(\ref{def:N_1_star_}) and Ineq.~(\ref{def:N_1_assumption}). Combining the last two chains of inequalities, we obtain that
\begin{align*}
P\left[D \geq  d^*\right]  \geq P\left[D > d^*\right] \geq 1-\varepsilon.
\end{align*}
Accordingly, with probability $1-\varepsilon$, no citizen will vote in the second round. Given this outcome, citizens in the first round do not want to change their sincere voting decision. On the one hand, all first-round citizens whose preferred alternative is $A$ are {content with their decisions as their preferred outcome is implemented}. On the other hand, all first-round citizens whose preferred alternative is $B$ would not obtain a better outcome by switching their vote towards $A$ in the first round, for this would only increase $d$. This completes the proof.

\end{proof}

		

		

\bigskip

\begin{proof}[Proof of Corollary \ref{cor:main_theorem}:]

Given the proof of Theorem \ref{thm:main_theorem}---see (\ref{def:N_1_star_})---, it follows immediately that $N_1^*$ increases if either $\varepsilon$ or $p_A-p_B$ decreases. We now focus on changes on $c$. From the proof of Lemma \ref{main_lemma}---see Eq. (\ref{eq:convergence})---, we know that $d^*(c)$ decreases as $c$ decreases. Since $N_1$ is increasing when $d^*$ is increasing, the claim holds.

\end{proof}

\bigskip

\begin{proof} [Proof of Theorem \ref{thm:welfare}:]
	
	Under AV, the average per-capita social cost of subsidizing is $f \cdot c$, where $f$ is the expected ratio of the AG size {to} the total number of the voters 
	\begin{equation}
	\label{eq:f_definition}
	f=\mathbb{E} \left[\frac{N_1}{N_1+N_2}\right].
	\end{equation}
  Since $N_2$ is a Poisson random variable with parameter $n_2$, we can easily obtain the following upper bound for $f$:
\begin{equation}
\label{ineq:n_2}
f=N_1\cdot \sum_{k=0}^{\infty}\frac{1}{N_1+k}\frac{n_2^k}{k!e^{n_2}}\leq N_1\cdot  \sum_{k=0}^{\infty}\frac{n_2^k}{(k+1)!e^{n_2}} = \frac{N_1}{n_2}\cdot  \left(1-\frac{1}{e^{n_2}}\right). 
\end{equation}
In particular, for a fixed $N_1$, we have $\lim_{n_2\rightarrow \infty} f =0$. Next, according to Theorem \ref{thm:main_theorem}, if $N_1\geq N_1^*(\varepsilon,c,p_A-p_B)$, the outcome will be fully determined by AG with probability $1-\varepsilon$. Therefore,
	\begin{align}
	\label{ineq:one}
	W^{AV}\geq (1-\varepsilon)\cdot \left(w_d(N_1, n_2) - c f \right) + \varepsilon \cdot (-c),
	\end{align}
	where $\varepsilon>0$ and $w_d(N_1, n_2)$ is the expected average welfare (in the entire population) obtained from the alternative implemented when members of AG,  a group of size $N_1$, vote sincerely. It is easy to verify that
		\begin{equation}
	\label{eq:welfare}
	w_d (N_1, n_2) = p_A  \cdot (1-z_d(N_1, n_2)), 
	\end{equation}
	with 
	\begin{equation}
	\label{eq:limit_welfare}
	\lim_{N_1 \rightarrow \infty} z_d(N_1, n_2)=0.
	\end{equation} 
	Hence, there is $\varepsilon^*>0$ such that for all $N_1\geq N_1^*(\varepsilon^*,c,p_A-p_B)$, we derive from Ineq. (\ref{ineq:one}) that
	\begin{align}
	\label{ineq:two}
	W^{AV}>  p_A  - cf + \delta (N_1, n_2),
	\end{align}
	where 
	\begin{equation}
	\label{limit:one}
	\lim_{N_1\rightarrow \infty} \delta(N_1, n_2)=0. 
	\end{equation}
	Finally, because $p_A-p_B>0$ and due to Ineq. (\ref{ineq:n_2}) and (\ref{limit:one}), {there must be $N_1^{**} (c)$, with $N_1^{**} (c)\geq N_1^*(\varepsilon^*,c,p_A-p_B)$,} and $n^*_2 (c)$ such that if $N_1\geq N^{**}_1(c)$ and $n_2 \geq n_2^{*} (c)$,
	\begin{align*}
	p_A  - cf + \delta (N_1,n_2) > p_A-c = W^{com}
	\end{align*}
	and
	\begin{align*}
	p_A  - cf + \delta (N_1,n_2) > \frac{1}{2}-\frac{2x}{N_1+n_2} \cdot c = W^{vol},
	\end{align*}
	where $x$ is the solution to Eq. (\ref{eq:simultaneous_voting}). In combination with (\ref{ineq:one}) and (\ref{ineq:two}), the latter two inequalities prove that 
	\begin{equation*}
	W^{AV} > \max \{W^{vol},W^{com}\}.
	\end{equation*}
\end{proof}

\pagebreak

\section*{Appendix B}


In this Appendix, we extend the properties of AV in two directions: first, we analyze the robustness of Corollary \ref{corollary:non_existence_result} when citizens with the same preferences use different strategies; second, we investigate the performance of Assessment Voting when there are more than two alternatives.

\subsection*{Multiple citizen types}

In the main body of the paper, we have assumed that all agents who preferred alternative $A$ to $B$ used the same strategy. In particular, in our analysis of game $\mathcal{G}^2(d)$, we considered that all citizens played according to one of two strategies: $\alpha_A$ for citizens whose preferred alternative is $A$ and $\alpha_B$ for citizens whose preferred alternative is $B$. In this section, we assume that citizens of type $A$ and $B$ may be of different (sub)types, which are given exogenously. 

More specifically, for a given integer $T\geq 1$, let $\mathbb{S}^T=\{(\rho_k)_{k=1}^T \vert \rho_1,\ldots,\rho_T\geq 0, \sum_{k=1}^T \rho_k=1\}$ denote the $T$-simplex. Then, we assume that there exist $\rho_A=(\rho_A^k)_{k=1}^{T^A} \in \mathbb{S}^{T^A}$ and $\rho_B=(\rho_B^k)_{k=1}^{T^B}\in \mathbb{S}^{T^B}$, with $T^A,T^B\geq 1$, such that any citizen $i$'s probability of being of (sub)type $t^k_A$ is equal to $p_A\cdot \rho^k_A$. We shall assume that citizens of different (sub)types may use different strategies, i.e, they may randomize between going to the ballot box or not, using different probabilities. Accordingly, we have $\alpha_{A,k}$, with $\alpha_{A,k}\in [0,1]$, denote the probability according to which citizens of type $t^k_A$ will turn out (and then vote for alternative $A$). In turn, $\alpha_{B,k}$ can be analogously defined for $B$-supporters. By the properties of the Poisson probability distribution, in the second round of AV, the number of citizens of each (sub)type $t^k_A$ is a Poisson random variable with parameter $n_2\cdot p_A \cdot \rho_A^k \cdot \alpha_{A,k}$, which we denote by $x_{A,k}$. Similarly, the number of citizens of each (sub)type $t^k_B$ is a Poisson random variable with average $n_2\cdot p_B \cdot \rho_B^k \cdot \alpha_{B,k}$, which we denote by $x_{B,k}$. We recall that $d^*(c)$ has been defined as the (minimal) threshold guaranteeing that if $d\geq d^*(c)$, no citizen will turn out in the second round of AV. We can prove the following result, which generalizes Corollary \ref{corollary:non_existence_result} to a setting with multiple citizen types.



\begin{proposition}\label{proposition:non_existence_result_multiple_types}
	Assume that there are $T^A$ (sub)types of $A$-supporters and $T^B$ (sub)types of $B$-supporters. For any cost $c$, with $0<c<1/2$, if $d\geq d^*(c)$, the only equilibrium is the no-show equilibrium. 
\end{proposition}
\begin{proof}
	
	Let $\mathbb{N}$ denote the set of non-negative integer numbers. The fact that the no-show strategy profile is an equilibrium is trivial, provided that $d^*(c)\geq 2$. To show that this is the unique equilibrium, we  distinguish two cases.
	
	\textbf{Case I: $T^A \geq 1$ and $T^B=1$}\nopagebreak
	
	For all voters of type $A$, regardless of their subtype, the indifference condition between turning out and abstaining is the following: 
\begin{align}
\label{indifference_condition}
2c & = \sum_{\left(k_1,..., k_{T^A}\right)\in \mathbb{N}^{T^A}}\prod\limits_{r=1}^{T^A}\frac{{x_{A,r}}^{k_r}}{e^{x_{A,r}}k_r!} \cdot \bigg(\frac {{x_B}^{\sum_{s=1}^{T^A}k_s+d}}{e^{x_B}(\sum_{s=1}^{T^A}k_s+d)!} + \frac {{x_B}^{\sum_{s=1}^{T^A}k_s+d+1}}{e^{x_B}(\sum_{s=1}^{T^A}k_s+d+1)!}\bigg) .
\end{align}
Nevertheless, by Ineq. (\ref{eq:aux_1})---see the proof of Proposition \ref{main_result}---, we obtain that for all $d\geq d^*(c)$ and all  $x_B\in \mathbb{R}_+$,
\begin{align}
\label{eq:ineq_case_A}
& \sum_{\left(k_1,..., k_{T^A}\right)\in \mathbb{N}^{T^A}}\prod\limits_{r=1}^{T^A}\frac{{x_{A,r}}^{k_r}}{e^{x_{A,r}}k_r!} \cdot \bigg(\frac {x_B^{\sum_{s=1}^{T^A}k_s+d}}{e^{x_B}(\sum_{s=1}^{T^A}k_s+d)!} + \frac {x_B^{\sum_{s=1}^{T^A}k_s+d+1}}{e^{x_B}(\sum_{s=1}^{T^A}k_s+d+1)!}\bigg)  \notag \\
<& \sum_{\left(k_1,..., k_{T^A}\right)\in \mathbb{N}^{T^A}}\prod\limits_{r=1}^{T^A}\frac{{x_{A,r}}^{k_r}}{e^{x_{A,r}}k_r!} \cdot 2c = 2c,
\end{align}
where the second inequality holds from the claim that 
\begin{equation}
\label{eq:claim}
\sum_{\left(k_1,..., k_{T^A}\right)\in \mathbb{N}^{T^A}}\prod\limits_{r=1}^{T^A}\frac{x_{A,r}^{k_r}}{e^{x_{A,r}}k_r!} = 1.
\end{equation}
Assuming Eq. (\ref{eq:claim}),  Eq. (\ref{indifference_condition}) does not have a solution, and hence there cannot be an equilibrium of game $\mathcal{G}^2(d)$ in which $A$-supporters are split into $T^A$ (sub)types and each (sub)type $t^r_A$ of citizen plays according to a totally-mixed strategy $x_{A,r}$. Finally, it only remains to prove Eq. (\ref{eq:claim}). We prove the claim by induction on $T^A$. The case $T^A=1$ holds directly from the properties of the Poisson probability distribution. Hence, assume that Eq. (\ref{eq:claim}) holds for some $T^A\geq 1$. Then,
\begin{align*}
\sum_{\left(k_1,..., k_{T^A+1}\right)\in \mathbb{N}^{T^A+1}}\prod\limits_{r=1}^{T^A+1}\frac{x_{A,r}^{k_r}}{e^{x_{A,r}}k_r!} &  = \sum_{k=0}^{\infty} \left[\sum_{\substack{ \left(k_1,..., k_{T^A+1} \right)\in \mathbb{N}^{T^A+1}, \\ k_{T^A+1}=k } } \left(\prod\limits_{r=1}^{T^A}\frac{x_{A,r}^{k_r}}{e^{x_{A,r}}k_r!} \cdot \frac{x_{A,T^A+1}^{k}}{e^{x_{A,T^A+1}}k!} \right) \right] \\
&  = \sum_{k=0}^{\infty} \left[\frac{x_{A,T^A+1}^{k}}{e^{x_{A,T^A+1}}k!} \sum_{ \left(k_1,..., k_{T^A} \right)\in \mathbb{N}^{T^A}  }\prod\limits_{r=1}^{T^A}\frac{x_{A,r}^{k_r}}{e^{x_{A,r}}k_r!} \right] \\
& = \sum_{k=0}^{\infty} \frac{x_{A,T^A+1}^{k}}{e^{x_{A,T^A+1}}k!} =1,
\end{align*}
where the penultimate equality holds by induction and the last equality holds due to the properties of the Poisson probability distribution.

	\textbf{Case II: $T^A \geq 1$ and $T^B\geq 1$}\nopagebreak
	
	Let us assume $T^A$ is given. We introduce further notation. Given $x_B=(x_{B,1}, ..., x_{B,T^B})$ and $k^B=(k^B_1,..., k^B_{T^B})$, we let  $P(x_B,k^B)$ denote the probability that, for each (sub)type $t^B_s$ ($s=1,\ldots,T^B$), there are exactly $k^B_s$ citizens of this (sub)type that vote, provided that citizens of type $t_B^s$ use strategy $\alpha_{B,s}$ (which leads to $x_{B,s}$). Because (sub)types are drawn independently, we obtain 
	\begin{equation*}
	P(x_B,k^B) = \prod\limits_{s=1}^{T^B}\frac{x_{B,s}^{k^B_s}}{e^{x_{B,s}}k^B_s!}.
	\end{equation*}
	Moreover, because of the multinomial theorem, we obtain that for all $m\geq 0$,
	\begin{equation}
	\label{eq:multinomial}
	\sum_{\substack{ \left(k^B_1,..., k^B_{T^B}\right)\in \mathbb{N}^{T^B}, \\ \sum_{s=1}^{T^B} k^B_s = m}} P(x_B,k^B) = \frac{\left(\sum_{s=1}^{T^B} x_{B,s}\right)^m}{ e^{\sum_{s=1}^{T^B} x_{B,s}}m!}.
	\end{equation}
For \textit{all} voters of type $A$, the indifference condition between turning out and abstaining is
	\begin{align}
	\label{indifference_condition_case_II}
	 2c & =  \sum_{\left(k^A_1,..., k^A_{T^A}\right)\in \mathbb{N}^{T^A}}\prod\limits_{r=1}^{T^A}\frac{{x_{A,r}}^{k^A_r}}{e^{x_{A,r}}k^A_r!} \cdot \left(\sum_{\substack{ k^B=\left(k^B_1,..., k^B_{T^B}\right)\in \mathbb{N}^{T^B}, \\ \sum_{s=1}^{T^B} k^B_s = \sum_{s=1}^{T^A} k^A_s + d}} P(x_B,k^B) + \sum_{\substack{ k^B=\left(k^B_1,..., k^B_{T^B}\right)\in \mathbb{N}^{T^B}, \\ \sum_{s=1}^{T^B} k^B_s = \sum_{s=1}^{T^A} k^A_s+d+1}} P(x_B,k^B)  \right) \notag \\
	 & =  \sum_{\left(k^A_1,..., k^A_{T^A}\right)\in \mathbb{N}^{T^A}}\prod\limits_{r=1}^{T^A}\frac{{x_{A,r}}^{k^A_r}}{e^{x_{A,r}}k^A_r!} \cdot \left(\frac{\left(\sum_{s=1}^{T^B} x_{B,s}\right)^{\sum_{s=1}^{T^A} k^A_s+d}}{ e^{\sum_{s=1}^{T^B} x_{B,s}}\left(\sum_{s=1}^{T^A} k^A_s+d\right)!} + \frac{\left(\sum_{s=1}^{T^B} x_{B,s}\right)^{\sum_{s=1}^{T^A} k^A_s+d+1}}{ e^{\sum_{s=1}^{T^B} x_{B,s}}\left(\sum_{s=1}^{T^A} k^A_s+d+1\right)!}  \right) \notag \\
	 &= \sum_{\left(k^A_1,..., k^A_{T^A}\right)\in \mathbb{N}^{T^A}} \prod\limits_{r=1}^{T^A}\frac{{x_{A,r}}^{k^A_r}}{e^{x_{A,r}}k^A_r!} \cdot \left(\frac{\sigma_B^{\sum_{s=1}^{T^A} k^A_s+d}}{ e^{\sigma_B}\left(\sum_{s=1}^{T^A} k^A_s+d\right)!} + \frac{\sigma_B^{\sum_{s=1}^{T^A} k^A_s+d+1}}{ e^{\sigma_B}\left(\sum_{s=1}^{T^A} k^A_s+d+1\right)!} \right) < 2c \notag,
	\end{align}
	where $\sigma_B:=\sum_{s=1}^{T^B} x_{B,s}$, the second equality follows from Eq. (\ref{eq:multinomial}), and the inequality follows from Ineq. (\ref{eq:ineq_case_A}) if $d\geq d^*(c)$. Because we have reached a contradiction, it must be that if $d\geq d^*(c)$, there cannot exist an equilibrium of game $\mathcal{G}^2(d)$ in which $A$-supporters are split into $T^A$ (sub)types and each (sub)type $t^r_A$ of citizen plays according to a totally-mixed strategy $x_{A,r}$, and in which $B$-supporters are split into $T^B$ (sub)types and each (sub)type $t^r_B$ of citizen plays according to a totally-mixed strategy $x_{B,r}$. This completes the proof.\footnote{The case in which some (sub)types play according to pure strategies can be proved analogously to the case considered here. We also note that although we have focussed on the case where $T_A$ and $T_B$ are finite numbers, the claim of Proposition \ref{proposition:non_existence_result_multiple_types} can be extended to the case where  $T_A$ or $T_B$ are infinite.}\end{proof}

\subsection*{Three or more alternatives}

In this section, we analyze the case of three or more alternatives. As already mentioned in the Introduction, this case has been analyzed by \cite{polborn} for one-round voluntary voting. We build on their approach to analyze AV. Specifically, we show that the negative result identified by Corollary \ref{corollary:non_existence_result} holds, regardless of the number of alternatives at hand. That is, we show that there is a threshold---which coincides with $d^*(\frac{c}{2})$ such that there is an equilibrium of the second-round voting game in which no citizen turns out, provided that the vote count difference in the first voting round is sufficiently large.\footnote{We cannot rule out the possibility that equilibria may also exist in which strategic voting occurs in the first voting round. One-round voting mechanisms also have the same drawback.}
 

Accordingly, suppose there is a set of $m$ alternatives $A_1, A_2, ... , A_m$, denoted by $\mathcal{A}$. Citizens are of one of $m!$ possible types $(A_1, A_2, ..., A_m), ... , (A_m, A_{m-1}, ... , A_1)$, where type $(A_{i_1}A_{i_2},..., A_{i_m})$ stands for the citizen whose most preferred alternative is $A_{i_1}$, the second most preferred alternative is $A_{i_2}$, and so on. Without loss of generality, we assume that there are $V_1,\ldots,V_m$ such that each  citizen $i$ derives a utility level $V_j$ if his $j^{th}$ best alternative wins. Without loss of generality we  impose the normalization $1=V_1 \geq V_2 \geq ... \geq V_m = 0$.

As in the case of two alternatives, we assume that the number of citizens of each type $(A_{i_1}, A_{i_2}, ... , A_{i_m})$ is distributed according to a Poisson random variable with parameter $p_{i_1,i_2,...,i_m}$. As for the solution concept, we assume that $(A_{i_1}, ... , A_{i_m})$-citizens who turn out vote for alternative~$A_{i_j}$, where $1\leq j \leq m-1$, with probability $p_{i_1,i_2,...,i_m}^{i_j}$. These probabilities are exogenously given and satisfy 
\begin{equation*}
\sum_{j=1}^{m-1}p_{i_1,i_2,...,i_m}^{i_j}=1.
\end{equation*}
In particular, we assume that citizens never vote for their least preferred alternative---this assumption generalizes sincere voting in a framework with at least three alternatives. 

Accordingly, we obtain that the number of voters in favor of alternative $A_j$ is distributed as a Poisson random variable (with parameter denoted by $\eta_j$), since it is a sum of independent Poisson random variables. As a tie-breaking rule, we consider that if there are $k$ alternatives with the same number of votes combined in the two voting rounds and if the remaining alternatives have strictly fewer votes, the alternative that wins is chosen among these $k$ alternatives, each alternative having probability $\frac{1}{k}$.

Next, suppose that alternative~$A_i$ has received $a_i$ votes in the first voting round. We can assume that $a_1\leq a_2 \leq ... \leq a_m$ without loss of generality. Finally, we let $\mathcal{G}^m(d)$ denote the modification of $\mathcal{G}^2(d)$, so that citizens can now vote for any of the $m$ alternatives in any voting round. We have the following proposition:

\begin{proposition}
	\label{three_alternatives}
	For any $c$ and any vector of votes $(a_1, a_2,...,a_m)$ after the first voting round, there is $d^{**}(c)$ large enough, such that if $a_m-a_{m-1}\geq d^{**}(c)$, the only equilibrium of game $\mathcal{G}^m(d)$ is the no-show equilibrium in the second round. 
\end{proposition} 
\begin{proof}
	The proof is based on an induction on $m$. The case $m=2$ is proven in Corollary~\ref{corollary:non_existence_result}. Now suppose that the claim of the proposition is true for the case of $m-1$ alternatives, and consider the case of $m$ alternatives. In particular, we will show that no citizen will vote for alternative~$A_1$ in any equilibrium of game $\mathcal{G}^m(d)$, where $A_1$ is the alternative that received the lowest number of votes in the first voting round. This means that instead of $m$ alternatives, it is as if there were only $m-1$ alternatives, $A_2,\ldots,A_m$. Because $a_m-a_{m-1}\geq d^{**}(c)$, we obtain by induction that the only equilibrium that survives is the no-show equilibrium.  

	
We distinguish two cases, which correspond to the cases in which one more vote in the second voting round in favor of alternative $A_1$ will make a difference in the final outcome. In both cases, we let $i$ be a citizen of type $(A_1, A_{i_2}, ... , A_{i_n})$. It will suffice to consider this type of citizen.\footnote{The argument of the proof can be easily adapted for all other types of citizens.}

	\textbf{Case I:} \textit{In the two voting rounds combined, alternative $A_1$ received exactly the same number of votes as each of  the alternatives of a given (non-empty) set $\mathcal{B}$, with all  alternatives in $\mathcal{A}\setminus(\mathcal{B}\cup \{A_1\})$ receiving strictly fewer total votes than those in $\mathcal{B}$.} \nopagebreak 
	
	In this case, with one additional vote in the second voting round, $A_1$ will win without ties. Accordingly, the expected gain that citizen $i$ derives from voting for $A_1$  in the second round is equal to
	\begin{equation*}
	H(\mathcal{B}):=1-\frac{1}{1+|\mathcal{B}|} \cdot \left(1+\sum_{j\in \mathcal{B}}V_{j}\right).
	\end{equation*}
	Let $x$ denote the total number of votes received by $A_1$ and alternatives from $\mathcal{B}$ in the two voting rounds combined. It is straightfoward to verify that $x\geq a_m$, because alternative~$A_m$ already has $a_m$ votes from the first round (the highest number among all alternatives). Then, the probability of having an alternative of set $\mathcal{B}$ winning the voting after both rounds (excluding $i$'s vote) is 
	\begin{align*}
	P_{equal}(\mathcal{B}):=\sum_{x=a_m}^{\infty}\left(\prod_{j \in \mathcal{B}\cup \{A_1\}}\frac{\eta_j^{x-a_j}}{e^{\eta_j}(x-a_j)!}\cdot P(x,\mathcal{A}\setminus (\mathcal{B}\cup \{A_1\})) \right),
	\end{align*}
	where $P(x,\mathcal{S})$ denotes the probability that alternatives in set $S$ \textit{all} receive strictly fewer votes than $x$. Let  $s$ denote the size of any arbitrary set $\mathcal{S}$. It is easy to verify the following: 
	\begin{equation}
	\label{eq:auxiliary_p}
	P(x, \mathcal{S}) = \sum_{\substack{(l_1,..., l_s)\in \mathbb{N}^s, \\ l_j+a_{s_j}<x, j=1,\ldots,s}}\prod_{r=1}^{s}\frac{\eta_{s_r}^{l_r}}{e^{\eta_{s_r}} l_r!}.	
	\end{equation}

		\textbf{Case II:} \textit{In the two voting rounds combined, alternative $A_1$ received one vote less than each of the alternatives of a given (non-empty) set $\mathcal{C}$, with all  alternatives in $\mathcal{A}\setminus(\mathcal{C}\cup \{A_1\})$ receiving strictly fewer total votes than those in $\mathcal{C}$.} \nopagebreak
		
			In this case, with one additional vote in the second voting round, there is a chance that $A_1$ will be chosen. Accordingly, the expected gain that citizen $i$ derives from voting in the second round in favor of $A_1$ is equal to
		\begin{equation*}
		F(\mathcal{B}):=\frac{1}{1+|\mathcal{C}|} \cdot \left(1+\sum_{j\in \mathcal{C}}V_j\right)-\frac{1}{|\mathcal{C}|} \cdot \sum_{j\in \mathcal{C}}V_j,
		\end{equation*}
		which is always a non-negative number since $\text{max}_{j\in \mathcal{C}}V_j \leq 1$.  Let $x+1$ now denote the number of total votes received by each of the alternatives in set $\mathcal{C}$ in the two voting rounds combined. That is, alternative~$A_1$ has received $x$ votes in both rounds combined, and it must be that $x\geq a_m$.  Then, the probability of having an alternative of set $\mathcal{C}$ winning the voting after both rounds (excluding $i$'s vote) is 
		\begin{align*}
		P_{low}(\mathcal{C}) := \sum_{x=a_m}^{\infty} \bigg(\frac{\eta_{1}^{x-a_1}}{e^{\eta_1}(x-a_1)!}\cdot \prod_{j\in \mathcal{C}}\frac{\eta_j^{x+1-a_j}}{e^{\eta_j}(x+1-a_j)!} \cdot P(x,\mathcal{A}\setminus(\mathcal{C} \cup \{A_1\}))\bigg),
		\end{align*}
		where $P(x,S)$ has been defined in Eq. (\ref{eq:auxiliary_p}).

\medskip
		
Finally, let $2^{\mathcal{A}}$ denote the power set of $\mathcal{A}$. Then, the indifference condition for citizen $i$ that equalizes the expected gain of voting for alternative~$A_1$ and the cost of voting is 
\begin{equation}
\label{eq:indifference_condition_three_alt}
c = \sum_{\mathcal{B}\in 2^{\mathcal{A}\setminus \{A_1\}}\setminus \emptyset}P_{equal}(\mathcal{B}) \cdot H(\mathcal{B})+\sum_{\mathcal{C}\in 2^{\mathcal{A}\setminus \{A_1\}}\setminus \emptyset}P_{low}(\mathcal{C})\cdot F(\mathcal{C}).
\end{equation}
By Ineq. (\ref{eq:aux_1})---see the proof of Proposition \ref{main_result}---, if $d\geq d^{**}(c)$ and for all $y\in \mathbb{R}_+$ and $k\geq0$, it holds that
\begin{equation}
\label{eq:aux_1_bis}
\frac{y^{k+d}}{e^{y}(k+d)!}<\frac{c}{2},
\end{equation}
If we now assume that $a_m-a_{1}\geq d^{**}(c)$, 
 then because $H(\mathcal{B})$ and $F(\mathcal{C})$ are at most one and all the events described in the calculations of $P_{equal}$ and $P_{low}$ are disjoint, we have: 
\begin{align*}
& \sum_{\mathcal{B}\in 2^{\mathcal{A}\setminus \{A_1\}}\setminus \emptyset}P_{equal}(\mathcal{B}) \cdot H(\mathcal{B})+\sum_{\mathcal{C}\in  2^{\mathcal{A}\setminus \{A_1\}}\setminus \emptyset}P_{low}(\mathcal{C})\cdot F(\mathcal{C}) \\
\leq &  \left(\sum_{\mathcal{B}\in 2^{\mathcal{A}\setminus \{A_1\}}\setminus \emptyset}P_{equal}(\mathcal{B})+\sum_{\mathcal{C}\in 2^{\mathcal{A}\setminus \{A_1\}}\setminus \emptyset}P_{low}(\mathcal{C}) \right)\leq  \left(\sum_{x=a_m}^{\infty}\frac{\eta_{1}^{x-a_1}}{e^{\eta_1}(x-a_1)!} \cdot \left( P_1(x)+P_2(x) \right)\right)
< c,
\end{align*}
where the strict inequality holds by Ineq. (\ref{eq:aux_1_bis}), and we also have that $$\sum_{x=a_m}^{\infty}P_1(x)\leq 1 \text{ and }\sum_{x=a_m}^{\infty}P_2(x)\leq 1.$$ That is, Eq. (\ref{eq:indifference_condition_three_alt}) cannot hold if $a_m-a_{1}$ is  above a certain threshold, which in fact coincides with $d^{*}(\frac{c}{2})$, and thereby is approximately four times bigger than $d^*(c)$. Letting $d^{**}(c)=d^{*}(\frac{c}{2})$ concludes the proof.\end{proof}


	

\end{document}